\documentclass[notitlepage,superscriptaddress,pra,twocolumn]{revtex4-1}

\usepackage{amsmath}
\usepackage{amsfonts}
\usepackage{amssymb}
\usepackage{graphicx}
\usepackage{sidecap}
\usepackage{braket}

\usepackage{blindtext}
\usepackage{subfigure}
\usepackage{lipsum}
\usepackage{fullpage}
\newcommand*\diff{\mathop{}\!\mathrm{d}}
\usepackage{appendix}
\usepackage{subfigure}
\usepackage{color}
\usepackage[dvipsnames]{xcolor}
\usepackage{physics}
\usepackage{hyperref}
\hypersetup{
	colorlinks=true,       
	linkcolor=VioletRed,  
    citecolor=JungleGreen,        
    filecolor=Orchid,      
    urlcolor=black           
}

 \usepackage{amsthm}
\newtheorem*{claim}{Claim}

\newtheorem*{dconj}{DRoP Conjecture}

\begin{document}
\title{Multidimensional super- and subradiance in waveguide quantum electrodynamics}
\author{Fatih Dinc}
\email{fdinc@stanford.edu}
\affiliation{Department of Applied Physics, Stanford University, Stanford, CA 94305, USA}
\affiliation{Perimeter Institute for Theoretical Physics, Waterloo, Ontario, N2L 2Y5, Canada}
\author{Lauren E. Hayward}
\affiliation{Perimeter Institute for Theoretical Physics, Waterloo, Ontario, N2L 2Y5, Canada}
\author{Agata M. Bra\'nczyk}
\affiliation{Perimeter Institute for Theoretical Physics, Waterloo, Ontario, N2L 2Y5, Canada}

\begin{abstract}
We study the collective decay rates of multi-dimensional quantum networks in which one-dimensional waveguides form an intersecting hyper-rectangular lattice, with qubits located at the lattice points. We introduce and motivate the \emph{dimensional reduction of poles} (DRoP) conjecture, which identifies all collective decay rates of such networks via a connection to waveguides with a one-dimensional topology (e.g. a linear chain of qubits). Using DRoP, we consider many-body effects such as superradiance, subradiance, and bound-states in continuum in multi-dimensional quantum networks. We find that, unlike one-dimensional linear chains, multi-dimensional quantum networks have superradiance in distinct levels, which we call multi-dimensional superradiance. Furthermore, we generalize the $N^{-3}$ scaling of subradiance in a linear chain to $d$-dimensional networks.
\end{abstract}

\maketitle
\section{Introduction}
Quantum networks composed of many nodes and channels~\cite{cirac1997quantum} hold remarkable promise for quantum computation~\cite{nielsen2010quantum,sipahigil2016integrated,arute2019quantum}, memory~\cite{lvovsky2009optical}, communication~\cite{shomroni2014all,zhou2013quantum},  sensing~\cite{degen2017quantum}, and ultimately, the quantum internet~\cite{kimble2008quantum}. Such networks could  be realized naturally in waveguide quantum electrodynamics (QED) systems, i.e., systems in which photons interact with quantum emitters inside waveguides \cite{sato2012strong,van2013photon,thompson2013coupling,goban2014atom,goban2015superradiance,schoelkopf2008wiring,ritter2012elementary,corzo2019waveguide,tiecke2014nanophotonic}. 

Given this potential for realizing quantum networks in waveguide QED systems, it is perhaps surprising that the theoretical study of many-body effects in waveguide QED has been mostly confined to single \cite{shen2005coherent,shen2007strongly,chang2007single,shen2009theory,baragiola2012n,zhou2017single,calajo2019exciting,1dchain} (see, for example, Fig.~\ref{fig:fig0}a) or coupled \cite{cheng2016coherent,zhou2013quantum,lu2014single,brod2016two,combes2018two,xu2013collective} waveguides with linear topologies. Nonetheless, a multi-dimensional quantum network, as idealized in Figs.~\ref{fig:fig0}b--c (or Fig.~1a of Ref.~\cite{kimble2008quantum}), provides compactness and higher connectivity within the network with its many nodes and channels. Yet, such a fruitful concept has been untouched in the waveguide QED literature so far, perhaps due to the inefficiency of current analytical and computational techniques for generalizing to larger dimensions. On the other hand, a systematic theoretical study of many-body effects in large multi-dimensional networks will pave the way for developing complex quantum networks, and ultimately, the quantum internet~\cite{kimble2008quantum}. With quantum computing becoming more of a reality with each passing day~\cite{arute2019quantum}, it is now a crucial time to address this problem.

\begin{figure}[b]
    \centering
\includegraphics[width=\columnwidth]{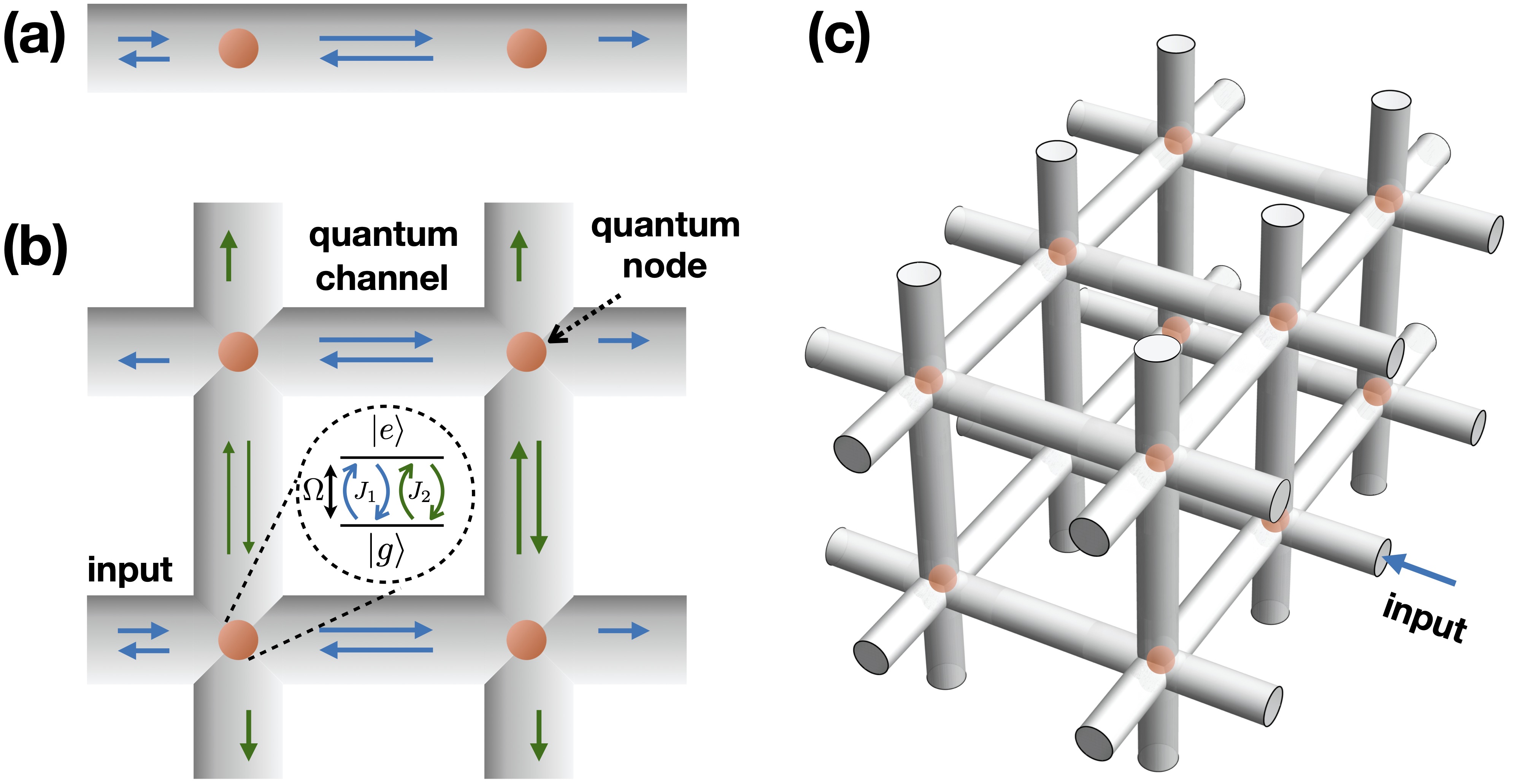}
 \caption{Schematic of quantum networks in (a) $d=1$,  (b) $d=2$, and  (c) $d=3$. 
 Quantum emitters, located at waveguide intersections, function as nodes, which are two-level systems with energy gap $\Omega$ that generate, store and process quantum information.
 Loss-less 1D waveguides form quantum channels that transport information, in the form of quantum light, between the nodes. In (b), blue arrows denote coupling to horizontal waveguides and green arrows denote coupling to vertical waveguides. 
 }\label{fig:fig0}
\end{figure}

In this paper, we introduce a strategy for investigating previously unexplored multi-dimensional collective phenomena in waveguide QED. Specifically, we highlight a connection between  multi-dimensional lattices and linear chains, and show that this connection can be used to compute the collective decay rates of multi-dimensional networks. We study collective phenomena such as super- and subradiance \cite{van2013photon,albrecht2019subradiant,asenjo2017exponential,kornovan2019extremely} and bound-states in continuum (BIC) \cite{facchi2016bound,facchi2019bound2} in these multi-dimensional systems\footnote{Here, the continuum refers to the the continuum states that are scattering energy eigenstates. The bound-states appear in this continuum at the distinct energy $E_k=\Omega$. See \cite{dinc2019exact} for a more thorough discussion on BIC in waveguide QED.}. In our investigations, we discover the concept of multi-dimensional superradiance, where, unlike the well-known phenomenon discussed by Dicke \cite{dicke1954coherence}, superradiance becomes partitioned for these multi-dimensional networks. As a consequence, we show that the dimensionality of BIC becomes smaller with increasing lattice dimension, presenting a trade-off between compact and efficient network design versus quantum memory capability. Moreover, we discover that the most subradiant decay rate in a $d$-dimensional hyper-cubic lattices scales as $N^{-3/d}$ for large $N$, generalizing the known $N^{-3}$ scaling in linear chains \cite{1dchain,albrecht2019subradiant,zhang2019theory}\footnote{Here, we note that stricter subradiant conditions have been found in the literature for specifically engineered qubit distances \cite{kornovan2019extremely}. Since we focus our analysis on parameter regimes that are within the vicinity of those that satisfy the super- radiance condition, we discuss multidimensional subradiance
based on $N^{-3}$ scaling. In general, for any $N^{-\alpha}$ scaling, the d-dimensional decay rate scales as$N^{-\alpha/d}$}.

\section{Theory and Results}
\subsection{System model}
We consider a $d$-dimensional lattice of intersecting loss-less waveguides with linear topologies \footnote{As defined in this paper, waveguides with linear topologies are one-dimensional networks; we note, however, that all waveguides considered in this paper are also one-dimensional in the sense that we parameterize the light in each waveguide in terms of the propagation direction and neglect transverse degrees of freedom.}, 
such as the ones illustrated in Fig.~\ref{fig:fig0}. Here, while the waveguides are shown to be intersecting, there are no geometrical restrictions in our model that the waveguides should be intersecting. Only the graph theoretical connections are important, and if the geometry can be deformed by keeping the graph topology intact, the theory does not change. For example, for $d=2$, one could have two separate layers of parallel waveguides, with the qubits being situated in between those two layers. In this way, there is no physical intersections between waveguides. 

In this model, each waveguide along the $n$th direction contains $N_n$ identical qubits, and each qubit couples to a waveguide in the $n$th dimension with the decay rate $\gamma_n$. The Hamiltonian for the entire system is $H = H_0 + H_I$, where $H_0$ includes the self-energies of the light and  qubits, and $H_I$ contains point-like interaction terms \footnote{The point-interactions lead to nearly-constant field amplitudes that change only at the lattice nodes (the point-like interaction assumption reduces what would be a set of differential equations to a set of linear equations).} located at the positions of the nodes \cite{shen2009theory}. The full expression for the Hamiltonian is presented in App. \ref{sec:appendixa}. We note that, in this Hamiltonian, we assume no non-radiative decay to outside of the system, no unwanted scattering at the intersection of the waveguides other than the one mediated by the qubits. We focus on this idealized model to provide a basis for the new physics that emerges in this confined lattice structure. If required, non-radiative decay could be considered by modelling loss modes as additional waveguides, as discussed in the literature \cite{Rephaeli:13}.

Of particular interest in waveguide QED systems are the reflection and transmission parameters, $r$ and $t$ respectively, and  excitation amplitudes $e$ corresponding to a plane wave with momentum $k$. Collectively, these quantities are known as  \emph{scattering parameters}, and can be found from the Hamiltonian $H$ by following the procedure outlined in Ref.~\cite{dinc2019exact}. 

Let us label each qubit using a $d$-dimensional lattice coordinate $\vec\sigma=(\sigma_1,\dots,\sigma_d)$.
If we concentrate on a particular direction $n \leq d$ with corresponding unit vector $\hat{n}$, we can define an adjacent qubit by the coordinate 
$\vec\sigma + a\hat{n} =
(\sigma_1, \dots, \sigma_n+a, \dots, \sigma_d)$, where $a$ is the lattice constant. The scattering parameters then satisfy the equations of motion (EoMs)
\begin{subequations}\label{eq:eom}
\begin{align} 
    t_{\vec\sigma + a\hat{n}}^{(n)} e^{-i k a} -t_{\vec\sigma}^{(n)} + i \sqrt{\gamma_n/2} e_{\vec\sigma} =0,\\
    r_{\vec\sigma + a\hat{n}}^{(n)} e^{i k a} -r_{\vec\sigma}^{(n)} - i \sqrt{\gamma_n/2} e_{\vec\sigma} =0, \\
    \sum_{n=1}^d \sqrt{\gamma_n/2} \left(t_{\vec\sigma}^{(n)}+r_{\vec\sigma}^{(n)}\right) - \Delta_k e_{\vec\sigma} =0,
\end{align}
\end{subequations}
which we derive in App.~\ref{sec:appendixa}. 
Within these equations,  
$\Delta_k = E_k-\Omega$ is the photon detuning energy, 
$E_k$ is the energy of the system, $k$ is the momentum of the photonic degree of freedom and $a$ is the lattice constant. Here we linearize the phase picked up by light propagating between two adjacent qubits such that $ka \simeq \Omega a=\theta$, which is accurate as long as time retardation effects inside the network are negligible. This assumption is valid in the Markovian regime, where the qubits are separated microscopically \cite{zheng2013persistent,dinc2019exact}. We note that, with our definition of scattering parameters, Eq.~(\ref{eq:eom}) is a high-dimensional generalization of the linear chain of qubits discussed in Eq.~(6) in Ref~\cite{1dchain}.

\subsection{Finding collective decay rates}
Now, we turn our attention to the \emph{collective decay rates} $\Gamma$, which contain information about the system's many-body structure. The collective decay rates are complex-valued, with their real components dictating the exponential decay of observables in time and their imaginary components capturing both the oscillatory behavior of the system as well as a characteristic frequency shift in energy levels~\cite{1dchain}. For each collective decay rate, there is a basis state in the single-excitation space that consists of some superposition of single-excited qubit states \cite{dinc2019exact}. For the scope of this paper, we do not focus on finding these states, but only the decay rates.

The collective decay rates can, in principle, be found by solving the EoMs and examining the poles of the scattering parameters, which yields a polynomial characteristic equation for $\Gamma$~\cite{dinc2019exact}. In a network of $N=\prod_{n=1}^d N_n$ qubits, the behaviour of the entire system is thus governed by a total of $(2d+1)N$ linear equations. For the special case of $d=1$, the boundary of the quantum network does not scale with $N$ as the linear chain has only one input and one output port. Hence, the equations of motion can be solved by eliminating all but the two boundary scattering parameters that are defined by the initial conditions. The rest of the scattering parameters can be found via back-propagation using transfer matrices. For this very special case, the scattering problem can be solved analytically and fairly efficiently for up to $N=500$ qubits using the transfer matrix method~\cite{dinc2019exact}. For $d\geq 2$, however, finding the scattering parameters by solving the EoMs becomes intractable for large~$N$. In this case, even if the system's internal degrees of freedom (corresponding to scattering within the system) can be eliminated, as is usually done in $d=1$ through transfer matrix methods \cite{dinc2019exact}, the number of external parameters (and, correspondingly, the number of equations to solve) scales with the size of the quantum network's boundary. 

To overcome some of these limitations, we introduce an idea that we call the \emph{dimensional reduction of poles} (DRoP), whereby we conjecture that there exists an effective mapping  between the collective decay rates (which are obtained from the poles of the scattering parameters \cite{dinc2019exact}) for the multi-dimensional network and waveguides with linear topologies.  The DRoP conjecture makes it possible to find the decay rates of previously intractable higher-dimensional quantum networks: first, one applies the DRoP conjecture to divide the multi-dimensional network into a subset of linear chains, and then one finds the collective decay rates of these linear chains using efficient transfer matrix methods \cite{dinc2019exact}, obtaining the decay rates of the quantum network in the process.
Through this method, one avoids performing calculations for $d$-dimensional, $N$-qubit networks directly, and instead works with a chain of size $\sim O(N^{1/d})$, for which it is possible to eliminate internal degrees of freedom.

\subsection{DRoP conjecture}
To motivate the conjecture, we begin with a simple example using decay rates obtained analytically by solving the EoMs directly. 
Consider a $d=2$ network with $N_n=2$ nodes and single emitter decay rates $\gamma_n$ along each direction $n=1,2$. In this case, there are four collective decay rates, given by
\begin{equation}\label{eq:2d}
\begin{aligned}
    \Gamma_1 &= \gamma_1(1 - e^{i\theta}) + \gamma_2(1 - e^{i\theta})\,,\\ 
    \Gamma_2 &= \gamma_1 (1 - e^{i\theta}) +\gamma_2(1 + e^{i\theta})\,, \\
    \Gamma_3 &= \gamma_1 (1 + e^{i\theta}) + \gamma_2(1 - e^{i\theta})\,,\\
    \Gamma_4 &= \gamma_1 (1 + e^{i\theta}) + \gamma_2 (1 + e^{i\theta})\,.
\end{aligned}
\end{equation} 
One can express these decay rates in terms of those corresponding to a $d=1$ waveguide along the direction $n$ with $N=2$ nodes, each with single emitter decay rate $\gamma_{n}$. For such a one-dimensional set-up, the decay rates are
\begin{align}\label{eq:1d}
    \Gamma_1^{(n)}=\gamma_{n}(1 - e^{i\theta})\,,\quad \Gamma_2^{(n)}=\gamma_{n}(1 + e^{i\theta})\,.
\end{align}
Here, the superscript in $\Gamma$ corresponds to the direction $n$, whereas the subscript distinguishes between distinct decay rates. Comparing Eqs.~\eqref{eq:2d} and~\eqref{eq:1d} reveals that the decay rates for the $d=2$ system can be written in terms of the dimensionless decay rates $z_i^{(n)} \equiv \Gamma_i^{(n)}/\gamma_n$ (with $i=1,2$ denoting the decay rates along $n=1,2$ direction) for the linear system as
\begin{equation*}
\begin{split}
    &\Gamma_1= \gamma_1z_1^{(1)} + \gamma_2 z_1^{(2)}, \quad \Gamma_2= \gamma_1 z_1^{(1)} + \gamma_2z_2^{(2)}, \\
    &\Gamma_3= \gamma_1 z_2^{(1)} + \gamma_2z_1^{(2)}, \quad \Gamma_4 = \gamma_1 z_2^{(1)} + \gamma_2 z_2^{(2)}.
\end{split}
\end{equation*} 
Remarkably, as we show throughout this paper, an analogous construction appears to hold for arbitrary $d$ and set of $N_n$'s. Based on this observation, we propose the following conjecture. 

\begin{dconj}
Consider a hyper-dimensional lattice with $N=\prod_{n=1}^d N_n$ qubits, where $N_n$ is the  number of qubits along direction $n$. Let $z_{i}^{(n)}=\Gamma_{i}^{(n)}/\gamma_n$ denote the dimensionless collective decay rates along the direction $n$, where $\gamma_n$ is the single emitter decay rate corresponding to the same direction, such that $1\leq n \leq d$ and $1 \leq i \leq N_n$. Then the complete set of collective decay rates belonging to the $d$-dimensional quantum network is 
\begin{equation}
        \Gamma = \left\{ \sum_{n=1}^d z_{s_n}^{(n)} \gamma_n
        \; \Big\vert s_n \in \left\{ 1, 2, \ldots, N_n  \right\} \right\}
\end{equation}
with $\vec s=\vec \sigma /a$ denoting the set of indices of the decay rates (with $|\vec s|=d$) and $|\Gamma|=N$ equal to the total number of qubits inside the network. Note that there are $N$ unique sets $\vec s$, each corresponding to a single collective decay rate $\Gamma_{\vec s}$.
\end{dconj}

We demonstrate in App. \ref{sec:appevidence} that decay rates obtained using DRoP match those found directly using the EoMs \footnote{We emphasize that while we focus on $\theta \approx m \pi$, where $m$ is a non-negative integer, in this paper, our demonstrations in App. \ref{sec:appevidence} show that DRoP is a more general phenomenon that holds for general $\theta$ values.}. We also find that DRoP is quite robust to random noise introduced to the single-qubit decay rates (see App.~\ref{sec:appendix_noise}). In the remainder of this paper, we use DRoP to discover new physics by probing regions inaccessible via EoM-motivated methods. We restrict our analysis to the physically relevant dimensions $d=2$ and $d=3$. 

\subsection{Multidimensional superradiance, subradiance, and BIC}
Guided by an intuition developed through studying the linear case~\cite{dinc2019exact}, we notice in our explorations that for $\theta \approx m\pi$ (with $m=0,1,\ldots$), the decay rates tend to cluster around certain regions of the complex plane. To understand the nature of this behavior, we focus on physical phenomena such as superradiance and subradiance (BIC for when $\theta=m\pi$) for waveguide QED systems.

Superradiance (subradiance) occurs when constructive (destructive) interference enhances (suppresses) spontaneous emission. 
Both physical phenomena are known to occur in a linear chain of qubits when $\theta=m\pi$ \cite{1dchain,zhou2017single,dinc2019exact}. When subradiance occurs,  $N-1$ decay rates converge to the origin, whereas superradiance is when one of the decay rates converges to $N\gamma$, where $\gamma$ is the single qubit decay rate. Out of $N$ possible first excited states, $N-1$ are dark states, i.e. states that have zero decay rate, owing to subradiance. Thus, one can construct the subspace containing the first excited states in terms of $N-1$ dark states, which do not couple to electromagnetic radiation, and one superradiant state that does. The collective system behaves as a two-level system between the one superradiant state and the ground state, which explains the Lorentzian shape of the transmission and reflection amplitudes, as discussed in \cite{zhou2017single}.

We find that in a $d$-dimensional quantum network, the dimension of the superradiant subspace is larger than one. As a result, the collective system no longer behaves as a qubit and is hence no longer described by Lorentzian transmission and reflection amplitudes. Consequently, the superradiant and subradiant states emerge differently than for the linear case. In a $d$-dimensional system, the DRoP conjecture predicts $\prod_n (N_n-1)$ subradiant states, with the rest showing superradiant features. We additionally find that, unlike in $d=1$, for $d>1$ superradiance is also partitioned. We define $n$-dimensional superradiance as the case where the decay rates of qubits along $n$ different directions are summed constructively. Then, there exist states with 1-, 2-, $\ldots$, and $d$-dimensional superradiance, a previously unobserved phenomena which we call \emph{multi-dimensional superradiance}. This partition of the superradiant behavior is illustrated for a small network in Fig.~\ref{fig:figure2} with $\theta=0.9999\pi$. We pick $\theta\neq \pi$ but close to $\pi$ to show the dimensionality of each cluster. We use this small network to validate our DRoP results with EoM-motivated methods. While partitioned nature of superradiance may be observable without the application of DRoP, the multi-dimensionality aspect cannot \footnote{Without DRoP, the partitioning of decay rates would be ambiguous, as in a single group, there are decay rates with different values. The partitioning is not made according to the decay rate values, but w.r.t. the physical meaning, e.g. which dimensions are being summed over. Without DRoP, there is no concept of summing over dimensions, hence no clear physical boundary for the partitions.}. The grouping into 1D, 2D and 3D superradiance is only possible through the DRoP conjecture and is consistent among large structures that cannot be accessed through the numerical algorithm discussed in App. \ref{sec:appendixc}. With DRoP, we find the origin of multi-dimensional superradiance, the dimension of each superradiant subspace as well as the strength of the corresponding $d$-dimensional superradiance. 

\begin{figure}
    \centering
    \includegraphics[width=\columnwidth]{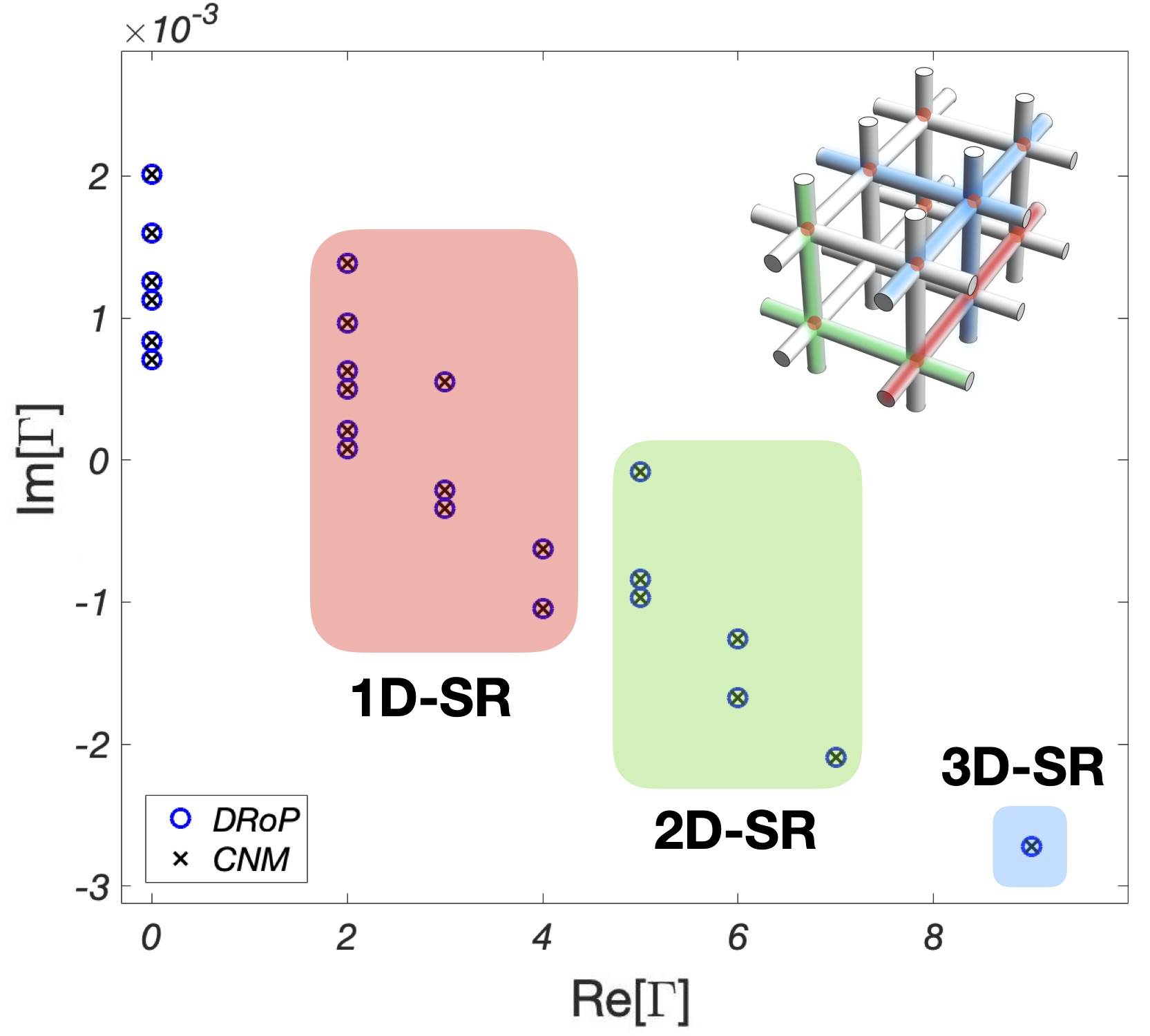}
    \caption{Collective decay rates (in units of $\gamma_1$) for a $d=3$ system. Notice the difference in scale between the $x$- and $y$-axes. We see that the decay rates cluster around the superradiant and subradiant values with vanishing imaginary part. Here, we used the parameters: $N_1\times N_2 \times N_3 = 2\times3\times4$, $\gamma_1=\gamma_2=\gamma_3$ and $\theta=0.9999\pi$.
    As in Figs.~\ref{fig:figure1a} and~\ref{fig:figure1b}, the black crosses correspond to results obtained from the numerical condition number method (CNM) discussed in App.~\ref{sec:appendixc}.
    }   \label{fig:figure2}
\end{figure}

Additionally, DRoP allows us to determine how a system's most-subradiant decay rate scales with its dimension. As pointed out in Refs.~\cite{1dchain,albrecht2019subradiant,zhang2019theory}, for a linear chain of qubits, the most subradiant decay rate scales as $N^{-3}$. Here, we find that for a $d$-dimensional quantum network, the corresponding scaling is $N_{\rm min}^{-3}$ with $N_{\rm min} \equiv \min_n [N_n]$. 
To find this scaling, we start by summing over the subradiant decay rates $\Gamma_{\rm sub}^{(n)}$ along each direction $n$ such that, overall, $\Gamma_{\rm sub}=\sum_n \Gamma_{\rm sub}^{(n)}$. From the literature, we know that for a linear chain of qubits, the scaling is $N^{-3}$ and thus $\Gamma_{\rm sub}^{(n)} \sim O(N_n
^{-3})$. The term with smallest $N_n$ dominates the summation, hence giving the $N_{\rm min}^{-3}$ expression.
For a hyper-cubic lattice, $N_{\rm min}=N^{1/d}$, leading to a subradiant decay rate that scales as $N^{-3/d}$. As expected, with increased dimensionality, subradiant behavior is washed out more gradually and the network couples to the continuum more strongly. Since subradiance has been considered crucial for quantum memory applications \cite{calajo2019exciting,dinc2019exact}, stronger design restrictions are expected to apply for higher dimensional quantum networks in comparison to linear chains.

For $\theta=m\pi$ and in the absence of non-radiative decay (which is the case considered in this paper), the subradiant states become BIC (or dark states). Consequently, the dimensionality of BICs is equal to the dimensionality of subradiant states, i.e. $\prod_n (N_n-1)$. The condition of BIC for a linear chain has been considered in Refs.~\cite{dinc2019exact,facchi2019bound2}.
As shown in App.~\ref{sec:appendixe}, the overall condition for BIC can be generalized to higher dimensions as
\begin{equation} \label{eq:bic}
    \sum_{c} p_c \, e_{\vec\sigma}^{(n)}=0 \text{  for all linear chains $c$} .
\end{equation}
Here, $p_c=1$ for even number of qubits along the chain and $\pm 1$ (alternating along the chain) for odd number of qubits along the chain. Since $\prod_n (N_n-1) \leq N-1$, the dimensionality of the BIC subspace in a $d>1$ dimensional quantum network is smaller than the one belonging to $N$ qubits in a linear chain. Therefore, it may be beneficial to use a lower-dimensional quantum network for memory applications. Here, again, design restrictions are expected to dictate the dimensionality of the quantum circuit being designed.

Finally, we note that our symmetry assumption is fundamental to the DRoP conjecture. There are three straightforward ways to violate the symmetry assumption: i) by adding noise to the decay rates $\gamma_n$ for each qubit, ii) by allowing qubits to have non-identical energy level separation $\Omega$, and iii) by varying the distance between qubits. In App. \ref{sec:appendix_noise}, we consider case i) and show that, as long as the noise added to the system is relatively small, DRoP can still be used to approximate the collective decay rates of the system, and such approximate values can later be used as seeds for the EoM-motivated search algorithms to find the exact decay rates.

\section{Conclusion}
Within this work, we have introduced the DRoP conjecture and illustrated its accuracy for a variety of examples  via analytical and numerical methods (App. \ref{sec:appevidence}). We have used DRoP to probe superradiance, subradiance and BIC in multi-dimensional quantum networks. We emphasize that, while EoM-motivated methods such as those discussed in App.~\ref{sec:appendixc} can provide some numerical information on the collective decay rates of small networks, our main results are derived from the analyticity that accompanies the DRoP conjecture and apply to multi-dimensional networks of arbitrarily large sizes. Previous research in waveguide QED mainly focused on linear structures and scattering parameters. While finding the scattering parameters efficiently in a large quantum network is still an open and important question, the DRoP conjecture opens the door for investigating multi-dimensional networks via their collective decay rates. 

We have focused our studies on cases where time retardation effects resulting from the inter-system photon propagation are neglected. In future work, it will be interesting to consider whether the DRoP conjecture holds in regimes where non-Markovian and time-delayed quantum coherent feedback effects become dominant~\cite{guimond2017delayed,calajo2019exciting,pichler2016photonic} and, in particular, to check whether recently-discovered super-superradiance effects for linear chains \cite{dinc2019non,sinha2020non} persist for higher-dimensional networks. We believe that the proof of the DRoP conjecture lies in deriving the matrix equation for the spontaneous emission dynamics (discussed in Ref.~\cite{dinc2019exact} for a linear chain), which would allow for studies of time evolution in multi-dimensional quantum networks. As collective decay rates are linked to  time evolution in waveguide QED~\cite{dinc2019exact}, it is possible that the dimensional reduction produced by DRoP could be present for the time evolution as well. Due to their potential to aid in the investigation of multi-dimensional networks, we expect DRoP to be useful for designing complex quantum networks for future quantum technologies. 

\section*{Acknowledgements}
FD acknowledges discussions with Diego Garcia and Jairo Rojas on Mathematica coding language. We also thank Gabriela Secara for help with preparing Figure \ref{fig:figure2}. 
Research at Perimeter Institute is supported in part by the Government of Canada through the Department of Innovation, Science and Economic Development Canada and by the Province of Ontario through the Ministry of Colleges and Universities.
We acknowledge the support of the Natural Sciences and Engineering Research Council of Canada (funding reference number RGPIN-2016-04135). 


\newpage
\appendix
\onecolumngrid
\section{Hamiltonian and energy eigenstates} \label{sec:appendixa}
In this section, we describe the Hamiltonian and the corresponding stationary states for a $d$-dimensional quantum network. The free and interaction Hamiltonians are given by
\begin{subequations}
\begin{align}
H_0 &= i \sum_{n=1}^d \sum_{m=1}^{\prod_{j\neq n} N_j} \int_{-\infty}^\infty \diff x \left( \psi_{n,m,L}^\dag(x) \frac{\partial}{\partial x} \psi_{n,m,L}(x) - \psi_{n,m,R}^\dag(x) \frac{\partial}{\partial x} \psi_{n,m,R}(x)  \right) +\Omega \sum_{\forall \vec \sigma} \ket{e_{\vec \sigma}} \bra{e_{\vec \sigma}},  \\
        H_I&= \sum_{n=1}^d \sum_{\forall \vec \sigma} \sqrt{\gamma_n/2} \Big[a^\dagger_{\vec \sigma} [\psi_{n,m,R}(\sigma_n) + \psi_{n,m,L}(\sigma_n) ] + h.c. \Big]. \label{eq:intham}
\end{align}
\end{subequations}
Here, $a^\dag_{\vec \sigma}$ is the excitation operator for the qubit whose position is given by the set $\vec{\sigma}=\{\sigma_1,...,\sigma_d\}$. $\gamma_n$ is the single qubit decay rate along the $n$th direction, $\psi^\dag_{n,m,L/R}(x)$ is the bosonic creation operator for the left/right moving photons at position $x$ in the $m$th waveguide along the $n$th direction, and $m=m(n,\vec \sigma)$. $\Omega$ is the energy separation of the qubit, $N_n$ is the number of atoms along direction $n$ and $\ket{e_{\vec \sigma}}=a_{\vec \sigma}^\dag \ket 0$ is the excited state for the $\vec \sigma$th qubit with $\ket 0$ being the superposition of the vacuum state and the ground state of all qubits. $\gamma_n$ is a constant and not a function of frequency, which is inline with the assumption that we are interested in energies $E_k \sim \Omega \pm O(\gamma_n)$, where $\gamma_n\ll \Omega$ \cite{shen2009theory}. Throughout the paper, we use natural units such that $\hbar=v_g=1$, where $v_g$ is the group velocity of photons inside the waveguides. 

Intuitively, if we were to consider time-evolution of this Hamiltonian, the first derivative in the free field Hamiltonian turns out to be a simple translation operator when we consider the interaction-free time-evolution. In that sense, the propagation is included in the spatial dependence of the free Hamiltonian. On the other hand, the interaction Hamiltonian derives the qubit excitations. In this paper, we are interested in stationary states, e.g. states that are energy eigenstates of the Hamiltonian. To find those states, we can construct a Bethe Ansatz as
\begin{equation}
        \ket{E_k} = \sum_{n=1}^d \sum_{\forall \vec \sigma} \int_{-\infty}^\infty \diff x \left(\phi_{n,\vec \sigma,R}(x) \psi_{n,m,R}^\dag(x) + \phi_{n,\vec \sigma,L}(x) \psi_{n,m,L}^\dag(x) \right) \ket 0 + \sum_{\forall \vec \sigma} e_{\vec \sigma}\ket{e_{\vec \sigma}} + \text{B.C}.. \label{eq:energyeigenstate}
\end{equation}
Here, $\phi_{n,\vec \sigma,R}(x)=t_{\vec \sigma}^{(n)} e^{ik(x-\sigma_n)} [\Theta(x-\sigma_n+a)-\Theta(x-\sigma_n)]$ and $\phi_{n,\vec \sigma,L}(x)=r_{\vec \sigma}^{(n)} e^{-ik(x-\sigma_n)} [\Theta(x-\sigma_n+a)-\Theta(x-\sigma_n)]$ are the piece-wise field amplitudes and $a$ is the lattice constant. At the boundary of the network, the field amplitudes include only one Heaviside function rather than two, meaning that the photon can radiate out of the system. This expression is a generalization of Eqs.~(3-4) in Ref.~\cite{1dchain}. $t_{\vec \sigma}^{(n)}$ and $r_{\vec \sigma}^{(n)}$ are, respectively, the transmission and reflection coefficients along the $n$th dimension belonging to the qubit $\vec \sigma$, whereas $e_{\vec \sigma}$ is the corresponding excitation coefficient. B.C. refers to the boundary terms of the photonic field, which can be hand-picked depending on the type of solution sought, due to the degeneracy of scattering eigenstates. Here we omit discussion the set of boundary conditions, as they do not have any effect on the collective decay rates. Applying the condition $H \ket{E_k}=E_k \ket{E_k}$, we show below that we obtain the equations of motion (EoMs) given in Eq.~\eqref{eq:eom} following the usual position space approach~\cite{shen2005coherent,1dchain}.

Applying the free Hamiltonian to the energy eigenstate, we find that
\begin{equation}
\begin{split}
    H_0 \ket{E_k}&= E_k \ket{E_k} - \Delta_k \sum_{\forall \vec \sigma} e_{\vec \sigma} \ket{e_{\vec \sigma}} \\
    &-i \sum_{n=1}^d \sum_{\forall \vec \sigma} \int_{-\infty}^\infty \diff x \,  t_{\vec \sigma}^{(n)} e^{ik(x-\sigma_n)} [\delta(x-\sigma_n+a) - \delta(x-\sigma_n)]  \psi_{n,m,R}^\dag(x) \ket 0  \\
    &+i \sum_{n=1}^d \sum_{\forall \vec \sigma} \int_{-\infty}^\infty \diff x \, r_{\vec \sigma}^{(n)} e^{-ik(x-\sigma_n)} [\delta(x-\sigma_n+a)-\delta(x-\sigma_n)] \psi_{n,m,L}^\dag(x) \ket 0  + \text{B.C.},
\end{split}
\end{equation}
where $\Delta_k = E_k-\Omega$ and $E_k = |k|$. For now, we do not put much emphasis on the boundary terms, although their shape will emerge at the end of our calculations. Applying the interaction Hamiltonian gives
\begin{equation}
    \begin{split}
       H_I \ket{E_k}&= \sum_{n=1}^d \sum_{\forall \vec \sigma} \sqrt{\gamma_n/2} \sum_{\sigma_n'} \left(\phi_{n,\vec \sigma,R}(\sigma_n') + \phi_{n,\vec \sigma,L}(\sigma_n') \right) \ket{e_{\vec \sigma}} \\
       &+ \sum_{n=1}^d \sum_{\forall \vec \sigma} \sqrt{\gamma_n/2} e_{\vec \sigma} [\psi_{n,m,R}^\dag(\sigma_n) + \psi_{n,m,L}^\dag(\sigma_n) ] \ket 0 + \text{B.C.}.
    \end{split}
\end{equation}
Now, shifting the indices, re-arranging some terms and using field continuity at the atomic positions, we find that
\begin{equation}\label{eq:appendixboundary}
    \begin{split} 
     H \ket{E_k}&= E_k \ket{E_k} +  \sum_{\forall \vec \sigma} \left( \sum_{n=1}^d  \sqrt{\gamma_n/2} \left[t_{\vec \sigma}^{(n)}+r_{\vec \sigma}^{(n)} \right] -\Delta_k e_{\vec \sigma} \right) \ket{e_{\vec \sigma}} \\
     &- i \sum_{n=1}^d \sum_{\forall \vec \sigma} \left[ t^{(n)}_{\vec\sigma + a\hat{n}}e^{-ika}-t^{(n)}_{\vec \sigma}+ i \sqrt{\gamma_n/2} e_{\vec \sigma} \right]  \psi_{n,m,R}^\dag(\sigma_n) \ket 0 \\
     &+i \sum_{n=1}^d \sum_{\forall \vec \sigma} \left[ r^{(n)}_{\vec\sigma + a\hat{n}}e^{ika}-r^{(n)}_{\vec \sigma}- i \sqrt{\gamma_n/2} e_{\vec \sigma} \right]  \psi_{n,m,R}^\dag(\sigma_n) \ket 0.
    \end{split}
\end{equation}
Here, $\vec\sigma + a\hat{n}$ is defined as $(\sigma_1,\dots,\sigma_n+a,\dots,\sigma_d)$ and we note that such terms originate from the index shifting in summations. Within this equation, the shape of the boundary terms arise analogous to the 1D case~\cite{1dchain,dinc2019exact}, and the boundary terms include incoming and out-radiating photonic components. For $\ket{E_k}$ to be an energy eigenstate, all the other terms in Eq.~\eqref{eq:appendixboundary} should be zero, which leads to the EoMs given in Eq.~\eqref{eq:eom} of the main text.

There is another more straightforward and elegant proof for deriving \emph{local} EoMs in waveguide QED systems with delta-function point interactions, for which we describe the strategy here. 
For a given system with many waveguides and qubits, one can divide the Hamiltonian into smaller pieces, with each piece containing a qubit and portion of all the waveguides that interact with it. These portions can be picked such that each is halved between the two adjacent qubits that are coupled to the same waveguide. Then, the Hamiltonian divides into sub-pieces such that 
\begin{equation}
    H = \sum_{Q} H_{Q},
\end{equation}
where $Q$ sums over all the qubits inside the system. Without loss of generality, $H_Q$ can be defined as
\begin{equation}
\begin{split}
      H_Q&=i \sum_W \int \diff x_W \left( \psi_{L,W}^\dag(x_W) \frac{\partial}{\partial x_W} \psi_{L,W}-\psi_{R,W}^\dag(x_W) \frac{\partial}{\partial x_W} \psi_{R,W} \right) \\
      &+ \Omega_Q \ket{e_Q}\bra{e_Q} + \sum_W \sqrt{\gamma_W/2} \Big[a^+_Q [\psi_{R,W}(Q_W) + \psi_{L,W}(Q_W) ] + h.c. \Big].  
\end{split}
\end{equation}
Here, $W$ stands for waveguides that interact with the qubit $Q$ and the upper and lower bounds of the integrals are not relevant, as they depend on the sub-division of waveguides into $H_Q$. As long as different sub-pieces are patched such that the photonic components are continuous at the patch points, no further equations of motion arise from the boundaries. Now, one can use the Bethe Ansatz approach to find the equation of motion around this single qubit as done in \cite{shen2005coherent,1dchain}, there are $2N_w^{(Q)}+1$ many equations resembling Eq. (\ref{eq:eom}) in the main text, with $N_w^{(Q)}$ being the number of waveguides interacting with the qubit $Q$. Bringing all together after patching, there are $\sum_Q (2 N_w^{(Q)}+1)$ EoMs for the whole system.

What makes this alternative proof more elegant is the fact that it does not use the specific geometry of the problem at hand. In fact, the local equations of motion are all the same for any type of waveguide QED system. The geometric properties of the system become important at the patching stage, where the inputs and outputs of patches should be properly defined to be continuous at the patch points (hence the irrelevance of space integral bounds). In a way, it is not the fundamental physics behind the EoMs that result in different emerging properties, such as the ones we discover in this paper, but rather the different phase relations that arise from patching in different geometries that lead to the different emergent phenomena.

\section{Evidence for DRoP conjecture}\label{sec:appevidence}

In the first comparison, we considered systems where $N_1,N_2,N_3 \leq 2$ for $d=3$, or $N_1,N_2\leq 3$ for $d=2$. In these cases, analytical results for collective decay rates can be found by directly solving Eq.~\eqref{eq:eom}. We compare these results to those obtained using DRoP and find that they agree exactly. We now illustrate the case for $N_1=N_2=3$ explicitly.

For a linear chain of three qubits, the dimensionless collective decay rates are~\cite{dinc2019exact}
\begin{subequations} 
\begin{align}
      z_1=\frac{\Gamma_1^{(1)}}{\gamma} &=\frac{1}{2}  \left(2+
   e^{2 i \theta} + e^{ i \theta}\sqrt{8  + e^{2 i \theta}}\right), \\
    z_2=\frac{\Gamma_2^{(1)}}{\gamma} &=\frac{1}{2} \left( 2+e^{2
   i \theta} -e^{ i \theta}\sqrt{8+ e^{2 i \theta}} \right), \\
    z_3=\frac{\Gamma_3^{(1)}}{\gamma} &=  \left(1-e^{2 i \theta}\right), 
\end{align}
\end{subequations}
with $\gamma$ corresponding to the single qubit decay rate. 
An analytical study of a $3 \times 3$ system with individual decay rates $\gamma_{1/2}$ shows that the collective decay rates of this higher-dimensional system are given by
\begin{subequations} 
\begin{align}
    \Gamma_1^{(2)} &=\frac{1}{2}  \left(2+
   e^{2 i \theta} + e^{ i \theta}\sqrt{8  + e^{2 i \theta}}\right)\gamma_1+\frac{1}{2}  \left(2+
   e^{2 i \theta} + e^{ i \theta}\sqrt{8  + e^{2 i \theta}}\right)\gamma_2, \\
    \Gamma_2^{(2)} &=\frac{1}{2}  \left(2+
   e^{2 i \theta} - e^{ i \theta}\sqrt{8  + e^{2 i \theta}}\right)\gamma_1+\frac{1}{2}  \left(2+
   e^{2 i \theta} + e^{ i \theta}\sqrt{8  + e^{2 i \theta}}\right)\gamma_2 \\
    \Gamma_3^{(2)} &=(1-e^{2i\theta}) \gamma_1 + \frac{1}{2} \left( 2+e^{2i\theta}+ e^{i\theta} \sqrt{8+e^{2i\theta}} \right) \gamma_2, \\
    \Gamma_4^{(2)} &=\frac{1}{2}  \left(2+
   e^{2 i \theta} + e^{ i \theta}\sqrt{8  + e^{2 i \theta}}\right)\gamma_1+\frac{1}{2}  \left(2+
   e^{2 i \theta} - e^{ i \theta}\sqrt{8  + e^{2 i \theta}}\right)\gamma_2, \\
    \Gamma_5^{(2)} &=\frac{1}{2}  \left(2+
   e^{2 i \theta} - e^{ i \theta}\sqrt{8  + e^{2 i \theta}}\right)\gamma_1+\frac{1}{2}  \left(2+
   e^{2 i \theta} - e^{ i \theta}\sqrt{8  + e^{2 i \theta}}\right)\gamma_2, \\
    \Gamma_6^{(2)} &=(1-e^{2i\theta}) \gamma_1+ \frac{1}{2} \left( 2+e^{2i\theta}- e^{i\theta} \sqrt{8+e^{2i\theta}} \right) \gamma_2,  \\
    \Gamma_7^{(2)} &=\frac{1}{2}  \left(2+
   e^{2 i \theta} + e^{ i \theta}\sqrt{8  + e^{2 i \theta}}\right)\gamma_1+(1-e^{2i\theta})\gamma_2,\\
    \Gamma_8^{(2)} &=\frac{1}{2} \left( 2+e^{2i\theta}- e^{i\theta} \sqrt{8+e^{2i\theta}} \right)\gamma_1 +(1-e^{2i\theta})\gamma_2,\\
    \Gamma_9^{(2)} &= (1-e^{2i\theta})\gamma_1+(1-e^{2i\theta})\gamma_2.
\end{align}
\end{subequations}
Re-writing these decay rates in terms of $z^{(1)}$, we obtain the effective mapping predicted by the DRoP conjecture such that
\begin{subequations} 
\begin{align}
    \Gamma_1^{(2)} &=z_1 \gamma_1+z_1\gamma_2, \\
    \Gamma_2^{(2)} &=z_2 \gamma_1+z_1\gamma_2, \\
    \Gamma_3^{(2)} &=z_3\gamma_1+z_1\gamma_2, \\
    \Gamma_4^{(2)} &=z_1\gamma_1+z_2\gamma_2, \\
    \Gamma_5^{(2)} &=z_2\gamma_1+z_2\gamma_2, \\
    \Gamma_6^{(2)} &=z_3\gamma_1+z_2\gamma_2,  \\
    \Gamma_7^{(2)} &=z_1\gamma_1+z_3\gamma_2,\\
    \Gamma_8^{(2)} &=z_2\gamma_1+z_3\gamma_2,\\
    \Gamma_9^{(2)} &= z_3\gamma_1+z_3\gamma_2.
\end{align}
\end{subequations}
Comparing the efficiency of both algorithms for finding the collective decay rates in this case, DRoP is $\sim 10^3$ times faster on a personal computer.
Similar analytical correspondence can be shown for a $3\times2\times2$, $4\times 3$ systems, both of which seem to be the boundary cases where Mathematica gives analytical results within less than a few hours on a standard personal computer. 

Next, for $N \sim 10$, finding the decay rates analytically by solving the EoMs is intractable. DRoP, on the other hand, yields analytical results. As an example, we consider a $4\times 4$ two-dimensional waveguide lattice. In Fig.~\ref{fig:figure1a}, we show analytical results using DRoP compared with results obtained numerically from the EoMs (App.~\ref{sec:appendixc} contains details on the numerical approach used). We  check all possible cases with $N_n\leq 3$ in $d=3$ and $N_n\leq4$ in $d=2$, and find that the direct numerical results agree with our analytical DRoP results within machine precision.

\begin{figure}[b]
    \centering
\includegraphics[width=16cm]{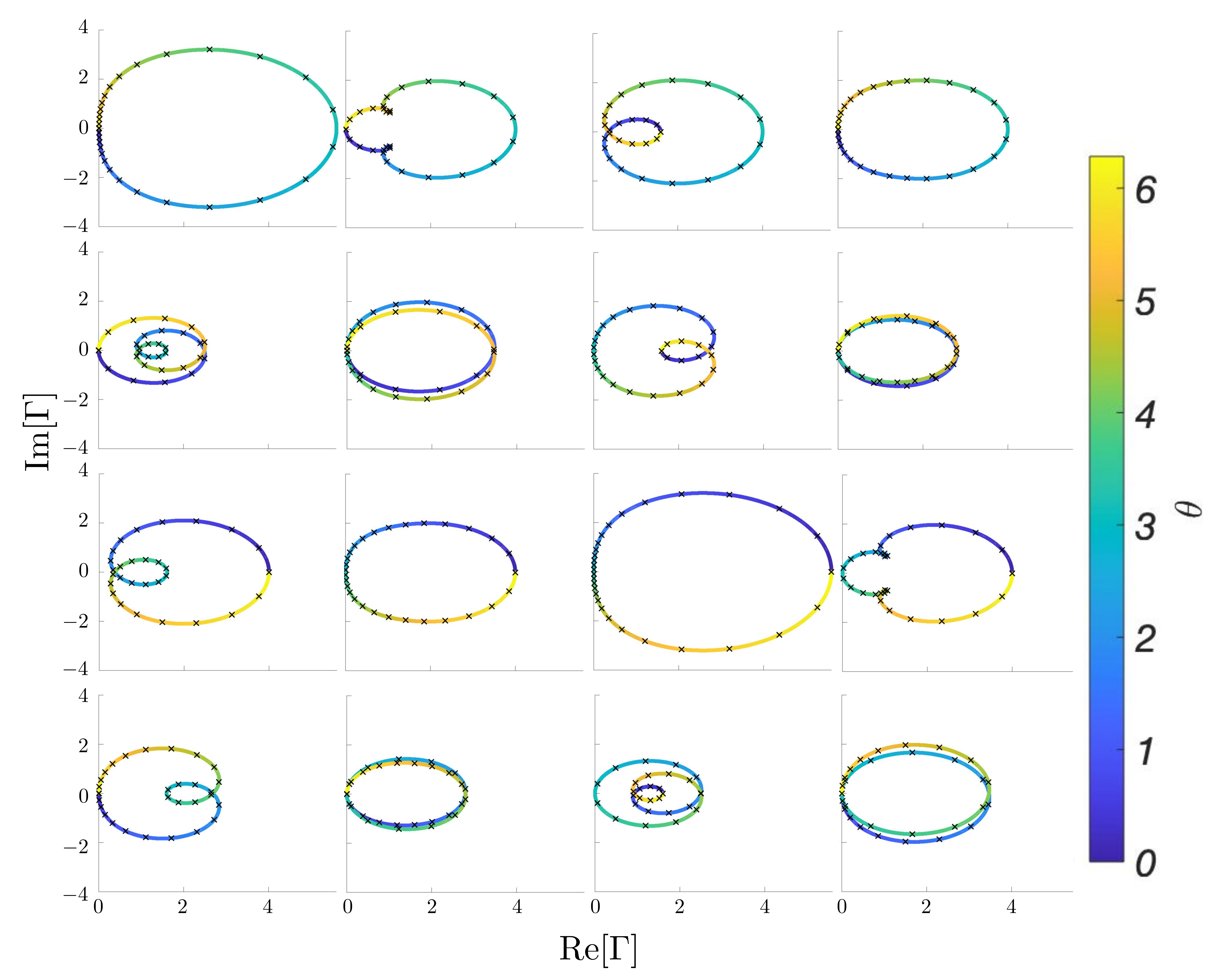}
    \caption{Collective decay rates  (in units of $\gamma_1$) for a $d=2$, $4\times4$ waveguide lattice ($d=2$ and $N_1=N_2=4$, with $\gamma_2/\gamma_1=0.4$). Each panel shows how particular decay rates behave in the complex plane for $\theta \in [0, \pi]$. The coloured solid lines correspond to results obtained analytically using DRoP while the black crosses correspond to direct results obtained from the EoMs using the numerical method discussed in App.~\ref{sec:appendixc}).}
    \label{fig:figure1a}
\end{figure}

\begin{figure}[t]
    \centering
\includegraphics[width=9cm]{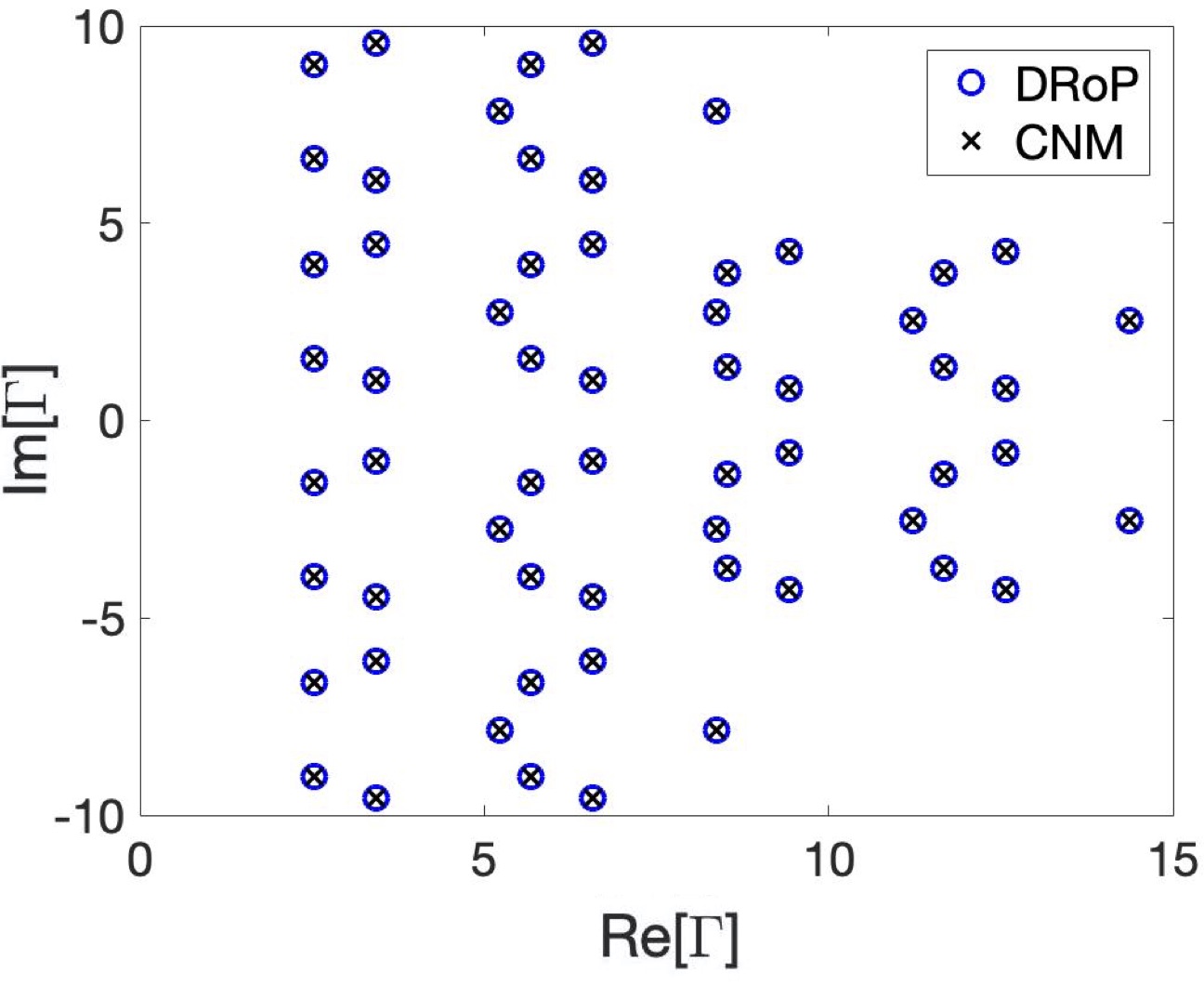}
    \caption{Collective decay rates $\Gamma_i$ (in units of $\gamma_1$) for a $d=3$, $5\times 3\times 4$ waveguide lattice ($d=3$, $N_1=5$, $N_2=3$, and $N_3=4$ for $\theta=\pi/2$ with $\gamma_2/\gamma_1=4$, $\gamma_3/\gamma_1=2$). The system has $60$ collective decay rates, all shown here.  The circles  correspond to results obtained numerically using DRoP while the black crosses correspond to direct results obtained numerically from the equations of motion using a numerical method known as the condition number method (CNM), as discussed in App.~\ref{sec:appendixc}.}
    \label{fig:figure1b}
\end{figure}

Analytical results using DRoP are possible even for large $N$, but analysis becomes cumbersome due to the fact that the analytical expressions for the decay rates end up having many branch cuts in the complex plane. We therefore turn to using DRoP numerically, and make comparisons for particular values of $\theta$. We check for various cases 
and again find that the decay rates found by numerically solving the EoMs directly and those found with DRoP (in combination with transfer matrix methods, see App.~\ref{sec:appendixc}) agree within machine precision. An example case for a three dimensional $5\times 3\times 4$ waveguide lattice is shown in Fig.~\ref{fig:figure1b}.

In all cases that we have checked, we find that the collective decay rates from DRoP match those found via EoM-motivated methods either exactly analytically or, when analytical comparison is not possible, to within machine precision.

We are using the EoM-motivated method to validate DRoP rather than using it to discover new physics, because the EoM-motivated method cannot access the regimes that we are discussing with DRoP. An immediate example of such an inaccessible regime is the multi-dimensional superradiance concept discussed in the main text. We know that for a linear chain, when $\theta=m\pi$, there are $N-1$ zero and one $N\gamma$ decay rates \cite{dinc2019exact}. From this, we arrive at the definition of multi-dimensional superradiance for any $N$ in the main text. Such a computation is not possible for the EoM-motivated method, as there will always be some $N'$ after which this method will be too computationally expensive.

\section{DRoP's robustness to noise} \label{sec:appendix_noise}

Here we probe DRoP's robustness to random errors that might occur during the fabrication of a quantum network. To account for situations where the individual decay rates $\gamma_n$ may not be exactly the same for all qubits along direction $n$, we append an additional numerical optimization step to DROP. To begin with, let the individual decay rates in Eq.~\eqref{eq:eom} be replaced with
\begin{equation}
    \gamma_n \to \gamma^{(n)}_{\vec \sigma}=\gamma_n (1 + \mathcal{N}(0,\epsilon_{\rm max}^2)).
\end{equation}
Here, $\mathcal{N}(\mu,\epsilon_{\rm max}^2)$ is the normal random distribution with mean $\mu$ and standard deviation $\epsilon_{\rm max}$ such that random noise is inserted into the decay rates at all atomic positions. Fig.~\ref{fig:noise} compares results predicted by DRoP, which approximates the noisy case as $\epsilon_{\rm max}=0$, to noisy parameters obtained via numerical methods. We can see that decay rates predicted by DRoP are good estimates  for the exact decay rates of the system based on the averages along one dimension $\gamma_n=\left< \gamma_{\vec \sigma}^{(n)} \right>_{\sigma}$. The estimates provided by DROP can then seed a minimum search algorithm (e.g. MATLAB's \texttt{fminsearch}) to reach the exact solutions more efficiently.

\begin{figure}
    \centering
    \includegraphics[width=8cm]{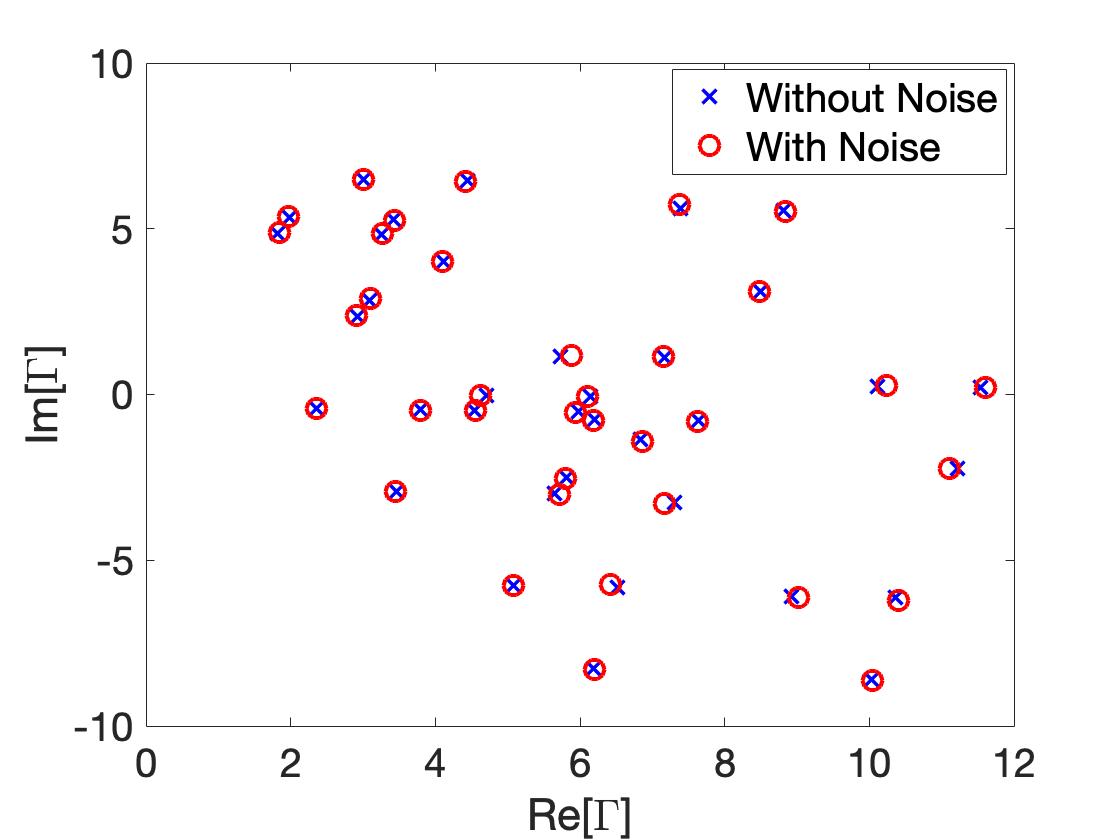}
    \caption{Comparison of results predicted by the DRoP conjecture and noisy parameters obtained via numerical methods (all in units of $\gamma_1$) with $\epsilon_{\rm max}=0.05$. 
    These results correspond to a $3\times 2 \times 6$ quantum network with $\gamma_2/\gamma_1=3$, $\gamma_3/\gamma_1=2$ and $\theta=0.65\pi$. }
    \label{fig:noise}
\end{figure}

\section{Methods used in the main text for finding the collective decay rates}\label{sec:appendixc}

Finding the scattering parameters by solving Eq.~(\ref{eq:eom}) for $d=1$ and a general $N$ is straightforward via the transfer matrix method \cite{1dchain}. Once the scattering parameters are known, the collective decay rates can be read off from the poles $\Delta^{(0)}=\{\Delta_k^{(0)}\}$  of the scattering parameters. However, for $d \geq 2$, solving for the scattering parameters become computationally intractable for even $N\sim 10$. As an example, for $\{d,N_1,N_2,N_3\}=\{3,2,3,4\}$, one needs to solve a set of $168$ coupled equations, which is not solvable on a standard personal computer within a time-span of few hours. 
As a result, in contrast to the 1-D case, solving for the scattering parameters is not a viable method to investigate the collective decay rates of a large multi-dimensional quantum network. For a $2$-D quantum network, solving Eq.~(\ref{eq:eom}) becomes analytically intractable for $N \sim O(10)$ and numerically intractable for $N \sim O(100)$ on a standard personal computer. 

Fortunately, if one is only interested in collective decay rates, instead of solving the linear system of equations multiple times, one can simply consider the matrix $A$, which contains the left hand side of Eq.~(\ref{eq:eom}), and find the set of poles $\Delta^{(0)}$, for which $A$ is singular. The decay rates can be found by rotation the poles in the complex plane via $\Gamma = 2i \Delta^{(0)}$ \cite{dinc2019exact}. 
\begin{claim}
The complete set of the poles of the scattering parameters is given by the values of $\Delta_k$ for which the matrix $A$ is singular.
\end{claim}
\begin{proof}
First, let us denote $\Delta^{(0)}$ as the $\Delta_k$ values for which $A$ is singular. Moreover, let us introduce the notation $A=A(\Delta_k)$. For this sketch, we use proof by contra-positive, meaning we shall show that if $A$ is non-singular, then $\Delta_k \notin \Delta^{(0)}$. 

If $A$ is non-singular, then it is invertible. Let us call denote the inverse matrix by $A^{-1}$. Then, the solution to the matrix equation can be given as
\begin{equation}
    x = A^{-1} b.
\end{equation}
If $A$ is invertible, then all entries of $A^{-1}$ are finite. Similarly, all entries of $b$ are finite by construction. Therefore, the scattering parameters are finite, since multiplication of two finite-dimensional matrices with finite entries results in a matrix with finite entries. By definition, if $\Delta_k \in \Delta^{(0)}$, then the scattering parameters should diverge by the existence of a pole, leading to a contradiction. Hence $\Delta_k \notin \Delta^{(0)}$. 
\end{proof}

This result simplifies our search for the poles, as we no longer need to solve the system of equations. Moreover, this method gives us important information regarding the maximum number of poles. Specifically, if the matrix $A$ is singular, then its determinant is zero. Now, the determinant of matrix $A$ is a polynomial with a degree of $N$, and hence there are at most $N$ poles. This result is in line with the findings of the literature so far~\cite{1dchain,dinc2019exact}, since the number of poles is expected to be bounded by the number of atoms inside the system.

To find the poles of $A$, one can solve the condition $\det(A)=0$. While the determinant algorithm provides useful insight, its implementation is cumbersome as the determinant of the $A$ matrix is a highly oscillating function of $\Delta k$. Fortunately, one can probe the singularity of a matrix by its eigenvalues since a matrix $A$ is singular if it has a zero eigenvalue. Thus, any pole $\Delta^{(0)}_k$ satisfies the property
\begin{equation}
     \Delta^{(0)}_k=\arg   \min_{\Delta k}  \abs{\text{eig}(A)}.
\end{equation}
We shall denote the algorithm using this approach as the ``eigenvalue method". 
While this method provides accurate results, it is slow. One can speed up the process by simply considering the condition number of the matrix $A$ instead of its whole eigenvalue spectrum. Then, the pole can be given by the so called ``condition number method" (CNM) as
\begin{equation}
     \Delta^{(0)}_k=\arg   \max_{\Delta k}  \text{cond}(A).
\end{equation}
This equation gives a pole depending on the seeding of the algorithm. By seeding the algorithm many times, one can find all $N$ poles. For the examples we consider in this paper, the poles are distinct, which is convenient for us when we show that DRoP works. On the other hand, there is no need for the poles to be distinct in the formulation of DRoP. The eigenvalue and condition number methods give the same results within numerical precision. Throughout this paper, we use the CNM to find the collective decay rates of a high-dimensional quantum network, since it is fastest. 

\begin{figure*}
    \centering
    \subfigure[]{\includegraphics[width=8cm]{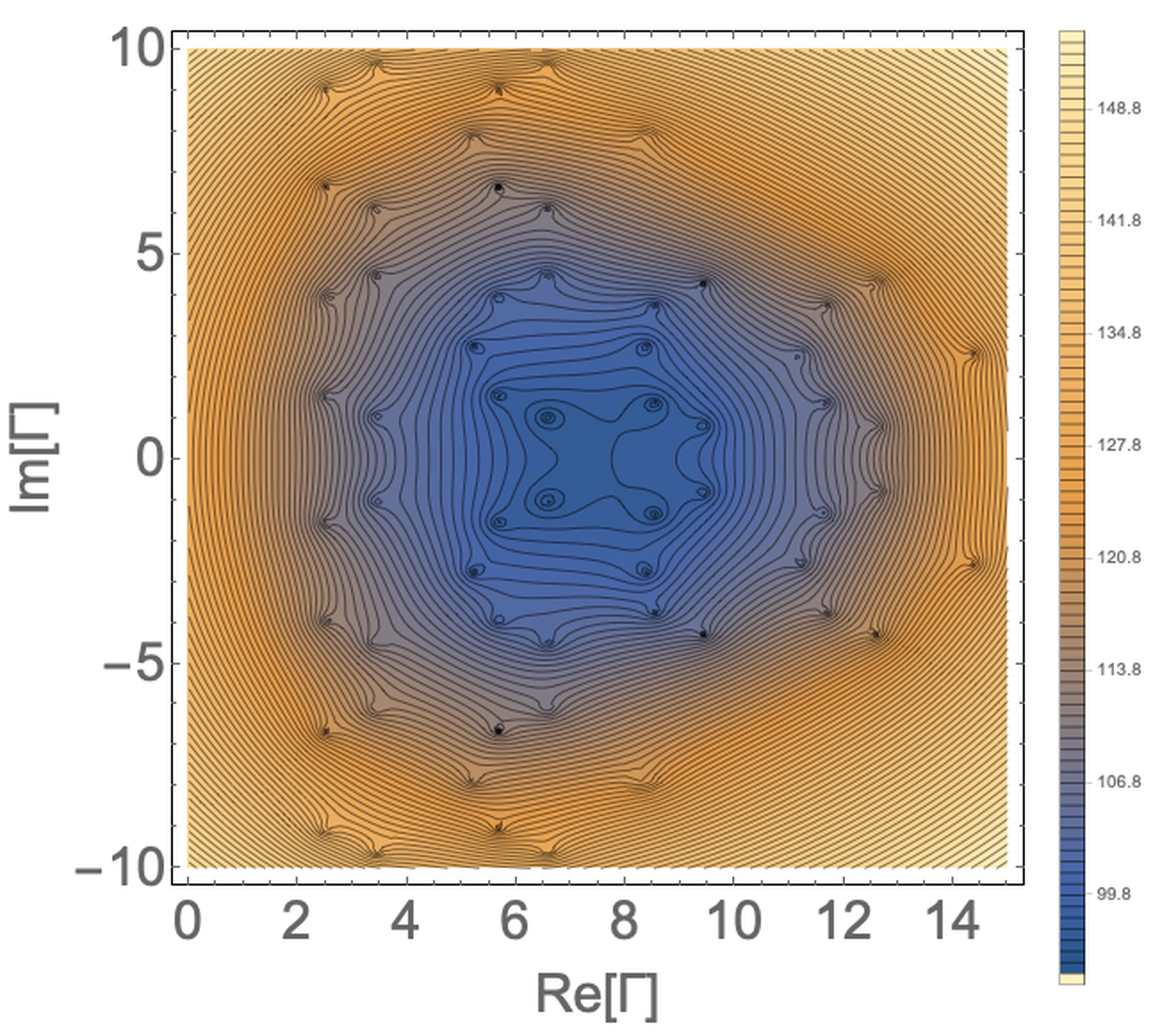}}
    \subfigure[]{\includegraphics[width=8cm]{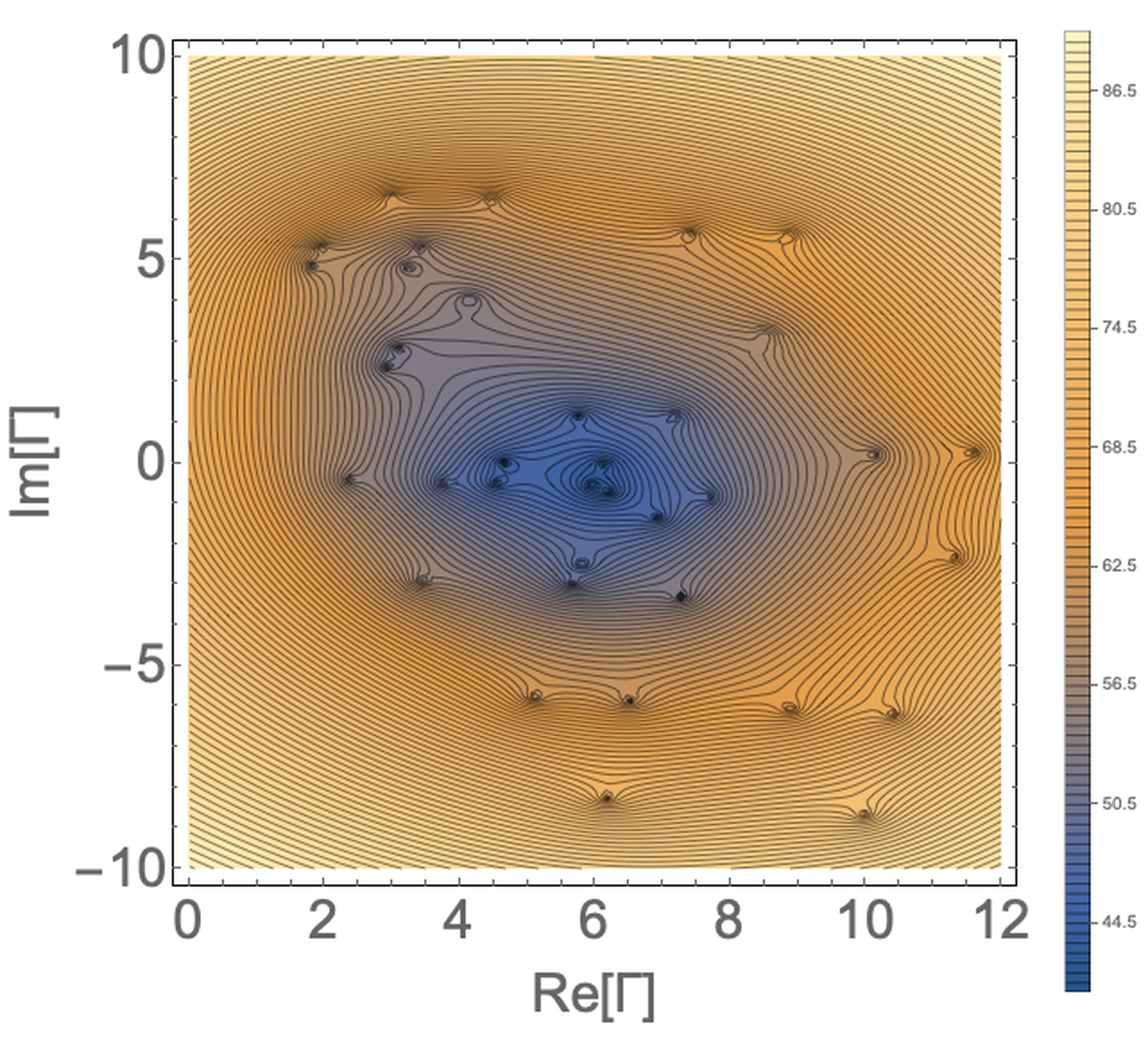}}
    \caption{The numerical values of $\log(|\det[A(\Gamma=2i\Delta_k)]|)$ (in units of $\gamma_1$) for two distinct cases: (a) $5\times 3\times 4$, $\gamma_2/\gamma_1=4$, $\gamma_3/\gamma_1=2$ and $\theta=0.5\pi$ with $\epsilon_{\rm max}=0$, (b) $3\times 2 \times 6$, $\gamma_2/\gamma_1=3$, $\gamma_3/\gamma_1=2$ and $\theta=0.65\pi$ with $\epsilon_{\rm max}=0.05$ (see App.~\ref{sec:appendix_noise} for the definition of $\epsilon_{\rm max}$) . Note that these numerical plots correspond to configurations given in Fig. \ref{fig:figure1b} and Fig. \ref{fig:noise}, respectively.}
    \label{fig:appendix}
\end{figure*}

In this work, we use the predictions made by the DRoP method to seed the condition number method, which is then used to find the poles of the quantum network. One might be concerned that the search algorithm finds a local minimum rather than a minimum corresponding to a zero. However, this concern can be addressed using the minimum modulus principle of complex analysis: Since $f(\Delta_k)=\det(A)$ is an analytical function of $\Delta_k$ (specifically here, a polynomial),
$|f|$ can only have local minima at the position of its zeros.

To check whether the results from the CNM are valid, we compare three different methods. In the first method, we consider the log-absolute determinant value of the matrix $A$ and show that there are indeed $N$ distinct minima, each corresponding to a pole, as plotted in Fig.~\ref{fig:appendix}. Second, we use the eigenvalue method to verify the results found by the CNM. Finally, we use Mathematica's \texttt{NRoots} function when possible to provide another check for our results. In all cases, the values found agree with the ones predicted by DRoP conjecture.

Having shown how the numerical approaches mentioned in the text work, let us now focus our attention on finding the collective decay rates of $N$ qubits in a linear chain efficiently. First, we present a compact two-equation system that includes all information about such decay rates. We then propose another method to find a polynomial characteristic equation of degree $N$. 

When $d=1$, the scattering problem has been solved via the transfer matrix method in Ref.~\cite{1dchain}. According to calculations performed in this reference, the set of equations describing the poles are given by
\begin{subequations}
\begin{align}
    \cos(\lambda) &= \cos(\theta) - \frac{\gamma}{2\Delta_k^{(0)}} \sin(\theta), \\
    (\Delta_k^{(0)} +i \gamma/2) \sin(N \lambda)&= \sin((N-1)\lambda) \Delta_k^{(0)} \exp(i\theta).
\end{align}
\end{subequations}
Here, $0 \leq Re[\lambda] \leq \pi$ is a complex parameter, $\gamma$ is the single qubit decay rate and $\Delta_k^{(0)}$ are the poles of the scattering parameters. This result shows that one can define dimensionless poles such that $z_p=\Delta_k^{(0)}/\gamma$ is the same for any $\gamma$, which means that the collective decay rates of a linear chain of atoms depend on $\gamma$ only linearly. Thus, we can define the dimensionless decay rates that describe $N$ qubit in a chain, regardless of the specific value of $\gamma$. 

In order to obtain analytical expressions for the decay rates, below we find a polynomial characteristic equation. To do so, we use the transfer matrix method, but employ a different approach than in Ref.~\cite{1dchain} for higher computational efficiency. The transfer matrix for a single unit cell, which includes a qubit and a propagation phase $\theta= \Omega a$, is
\begin{equation}
    \begin{pmatrix}
    t_{j-1} \\
    r_{j-1}
    \end{pmatrix}
    = 
    S
        \begin{pmatrix}
    t_{j} \\
    r_{j}
    \end{pmatrix}
    \implies
    \begin{pmatrix}
    t_{j-1} \\
    r_{j-1}
    \end{pmatrix}
    = 
    \begin{pmatrix}
   \left(1 + i \chi_k \right) e^{-i\theta} & i \chi_k e^{i\theta} \\
   -i \chi_k e^{-i\theta} &  \left(1 - i \chi_k \right) e^{i\theta}
    \end{pmatrix}
        \begin{pmatrix}
    t_{j} \\
    r_{j}
    \end{pmatrix},
\end{equation}
where $\chi_k = \gamma/(2\Delta_k)$. Here, $t_{j}$ and $r_j$ are the transmission and reflection coefficients for the $j$th atom. By construction, $t_0=1$ and $r_N=0$. Then, one can relate the output field amplitudes as
\begin{equation}
        \begin{pmatrix}
    1 \\
    r_0
    \end{pmatrix}
    = 
    S^N
        \begin{pmatrix}
    t_N \\
    0
    \end{pmatrix}.
\end{equation}
From this relation, one can find the final transmission coefficient as
\begin{equation}
    t_N = \frac{1}{(S^N)_{11}}.
\end{equation}
The characteristic polynomial describing the poles of the system is
\begin{equation}
    (S^N)_{11}(\chi_k^{(0)})=0,
\end{equation}
where $(S^N)_{11}(\chi_k^{(0)})$ is a polynomial with degree $N$ in terms of $\chi_k^{(0)}$. Note that once $\chi_k^{(0)}$ are known, $\Delta_k^{(0)}$ can be easily found.

\section{BIC in multi-dimensional quantum networks} \label{sec:appendixe}
In this Appendix, we prove the Eq. (\ref{eq:bic}) in the main text. Let us start by adding the first two expression in Eq.~(\ref{eq:eom}) such that
\begin{align} 
    t_{\vec\sigma + a\hat{n}}^{(n)} e^{-i\theta} + r_{\vec\sigma + a\hat{n}}^{(n)} e^{i\theta} =t_{\vec \sigma}^{(n)} + r_{\vec \sigma}^{(n)}.
\end{align}
This expression represents the emergence of wave-function continuity of the photonic field at the atomic positions. Now since, by the definition of BIC, the photonic field is zero outside the system, the sum of field amplitudes is always zero at the atomic positions by this continuity. Hence, from Eq.~(\ref{eq:eom}c), we find the first condition of BIC, namely
\begin{equation}
    \Delta_k = 0 \implies E_k = \Omega.
\end{equation}
All bound-states have the energy $E_k=\Omega$. 
Applying the second condition $\Omega L = m \pi$, where $m$ is a non-negative integer, we obtain the set of equations
\begin{subequations}
\begin{align} 
    (-1)^m t_{\vec\sigma + a\hat{n}}^{(n)}  -t_{\vec \sigma}^{(n)} + i \sqrt{\gamma_n/2} e_{\vec \sigma} =0,\\
   (-1)^m r_{\vec\sigma + a\hat{n}}^{(n)}  -r_{\vec \sigma}^{(n)} - i \sqrt{\gamma_n/2} e_{\vec \sigma} =0.
\end{align}
\end{subequations}
These equations are decoupled in each direction, and along a single dimension they mirror the 1D equations of motion. Therefore, following Ref.~\cite{dinc2019exact}, we find that the condition of BIC is simply the 1D conditions applied along each line inside the network, which leads to Eq.~(\ref{eq:bic}) of the main text.

\end{document}